\documentclass[pra,aps,twocolumn,superscriptaddress,showpacs]{revtex4}

\usepackage{graphicx,amsmath,amstext,amssymb,amsthm,bm,mathrsfs,amsbsy,latexsym,dsfont, array, layout,,mathrsfs,color}

\newcommand{\one}{\mathds{1}}

\newtheorem{thm}{Theorem}
\newtheorem{cor}[thm]{Corollary}

\newtheorem{defn}{Definition}

\begin{document}

\title{Entanglement Sharing Protocol via Quantum Error Correcting Codes}
\author{Ran Hee \surname{Choi}}
\affiliation{Institute for Quantum Science and Technology, University of Calgary, Alberta T2N 1N4, Canada}
\author{Ben Fortescue}
\affiliation{Institute for Quantum Science and Technology, University of Calgary, Alberta T2N 1N4, Canada}
\affiliation{Department of Physics, Southern Illinois University, IL, USA}
\author{Gilad Gour}
\affiliation{Institute for Quantum Science and Technology, University of Calgary, Alberta T2N 1N4, Canada}
\affiliation{Department of Mathematics and Statistics, University of Calgary, Alberta T2N 1N4, Canada}
\author{Barry C. Sanders}
\affiliation{Institute for Quantum Science and Technology, University of Calgary, Alberta T2N 1N4, Canada}

\begin{abstract}
We introduce a new multiparty cryptographic protocol,
which we call `entanglement sharing schemes',
wherein a dealer retains half of a maximally-entangled bipartite state
and encodes
the other half into a multipartite state that is distributed among multiple players.
In a close analogue to quantum secret sharing,
some subsets of players can recover
maximal entanglement with the dealer whereas other subsets can recover no entanglement (though they may retain classical correlations with the dealer).
We find a lower bound on the share size for such schemes and construct two non-trivial examples based on Shor's
$[[9,1,3]]$ and the $[[4,2,2]]$ stabilizer code; we further demonstrate how other examples may be obtained from
quantum error correcting codes through classical encryption.
Finally, we demonstrate that entanglement sharing schemes can be applied to characterize leaked information in quantum ramp secret sharing.
\end{abstract}


\maketitle

\section{Introduction}
Secret sharing is a well-studied multi-party cryptographic protocol
in the classical regime~\cite{Sha79,Bla79}
and the quantum regime,
which can correspond to either sharing classical messages via quantum channels~\cite{HBB99}
or sharing arbitrary quantum states~\cite{CGL99}
known as CQ and QQ versions, respectively~\cite{MS08};
we focus on QQ secret sharing here.
Multi-party cryptographic protocols typically are described as involving one party (the dealer)
possessing some information (the secret) which they distribute among the other parties (the players) in such a way that some
subsets of players can completely recover the secret, whereas other subsets of players are denied any knowledge
of the secret.
Secret sharing has numerous applications in cryptography including
secure multi-party computation, effecting the Byzantine agreement.
threshold cryptography, access control, and generalized oblivious transfer~\cite{Bei11}.
In the quantum regime, where the secret can be quantum information, quantum secret sharing (QSS) could  be applied
to quantum generalizations of applications for secret sharing,
such as distributed quantum computation~\cite{NMI01} or authorizing one party to access encrypted quantum communication.

Here we develop a theory of entanglement sharing schemes (ESS),
which is based on QSS including one dealer and
multiple players,
but we are concerned not with recovery of classical or quantum information
but rather with recovering shared entanglement between
the dealer and the players.
This secure protocol for sharing entanglement is important
as entanglement is a valuable communication resource when
only classical or degraded quantum channels are available between dealer and players after the initial sharing stage.
Entanglement is a crucial consumable resource for quantum communication
tasks such as quantum teleportation~\cite{BBC93}, super-dense coding~\cite{BW92}
and device-independent quantum key distribution~\cite{ABG07,MPA11}.

In entanglement
 sharing, the dealer, rather than encoding an arbitrary quantum state as in QSS, encodes half of a bipartite maximally-entangled state (MES), which is then
distributed to multiple players.
This distribution meets the requirement that the entanglement can be perfectly recovered between
the dealer and certain ``authorized'' subsets of players (the resultant shared entanglement being then usable as a resource).
Futhermore no entanglement
can be obtained between the dealer and other ``unauthorized'' subsets of players.
The ``access structure'' is the collection
of all authorized subsets, and the ``adversary structure'' is the collection of all unauthorized subsets.

The ESS, QSS and quantum error correcting coding (QECC) are interrelated as we shall discover in this paper.
Any QSS scheme can be modified to serve as an ESS,
and all ESSs can be derived from quantum error correcting codes.
In symbolic language we can write $\text{QSS}$.
Moreover quantum ramp secret sharing protocols (QRSSs) contain ESSs
and are contained in protocols corresponding to quantum error correcting codes.
Symbolically we can write
$$\widetilde{\text{QSS}}\subset\text{ESS}\subset\widetilde{\text{QRSS}}\subset\widetilde{\text{QECC}}$$
for the $\widetilde{}$ symbol denoting that the indicated scheme is actually its adaptation to ESS.

Entanglement sharing can be implemented using any QSS scheme: the dealer simply chooses the shared information to be half of an MES and proceeds as normal; the requirement for unauthorized subsets of the QSS scheme to have no information about the secret will ensure no shared entanglement with the dealer.  However, as we find, the differing requirements for ESSs
(in particular the ability for unauthorized subsets to share classical correlations with the dealer) allow for distinct schemes to be constructed.  We present examples of such schemes in later sections, and discuss how the use of classical encryption allows them to be constructed from error-correcting codes which may not naturally allow for such schemes.

\indent Finally we discuss the conceptual use of ESSs in the context of sharing arbitrary quantum states at a lower cost. In QSS
possessing perfect secrecy against information leakage, the size of shares allocated to each player must be at as large as the size of the secret~\cite{CGL99,GOT00}.
This limitation on the size of shares can impose large communication and storage costs.
QRSS is cheaper in the size of shares than quantum secret sharing, but it leaks some information to unauthorized sets of players, which are denoted as intermediate sets~\cite{OSIY05}. As the information leakage can compromise the secrecy of protocols, it is important to characterize the leaked information
and prevent the intermediate structure from learning any valuable information.\\

\section{Entanglement Sharing Schemes}
Our general form of an ESS is as follows:
suppose that a dealer~D initially prepares a quantum system
 in an MES $|S\rangle$.
 She retains half of the system, and encodes the other half into the shares of
 $n$ players $\text{P}=\{\text{P}_1,\text{P}_2,\dots,\text{P}_n\}$. Every subset of players must be either an authorized set or an
 unauthorized set. We denote as the recoverability condition that any authorized set of players can reconstruct $|S\rangle$ fully using
local operations (with respect to the player/dealer divide; they may perform joint quantum operations within their subset), and as the secrecy condition that any unauthorized set cannot share any entanglement with the dealer.\\
\indent We note that in this scheme every access structure is monotone i.e.\ any set of shares including an authorized set
 is also authorized. Also the complement of any authorized set must be unauthorized, due to
{\it monogamy of entanglement}~\cite{CKW00,OV06},
which states that if two systems are maximally entangled they cannot be
 entangled with any third system.  This can be expressed, for example, via the inequality
\begin{equation}
	E(\rho^{XY})+E(\rho^{XZ})\leq E(\rho^{X(YZ)})
\end{equation}
for any composite system of $X$, $Y$ and $Z$. This inequality holds for some entanglement measures such as the one-way distillable
entanglement and the squashed entanglement~\cite{KW04, KGS12}. \\
\indent Let $X$ be the dealer, $Y$ be an authorized set and $Z$ be the complement of $Y$. As $\rho^{XY}$ can be transformed into $\rho^{X(YZ)}$
by LOCC (local operations and classical communication) with respect to the dealer/player divide, and vice versa, $\rho^{XY}$ must have the same
amount of entanglement with respect to this divide as $\rho^{X(YZ)}$. Therefore, the complement of any authorized set is an unauthorized set,
i.e., $E(\rho^{XZ})=0$.
It similarly follows that, due to the monogamy of entanglement, an access structure in an
ESS cannot have two disjoint
subsets, since both such subsets could simultaneously recover maximal entanglement with the dealer.  This is analogous to quantum secret
sharing not allowing an access structure to have two disjoint subsets due to the no-cloning theorem.

\subsection{Properties}
\label{subsec:properties}
We derive the upper bound of the recoverable entanglement using the {\it non-lockability} of the relative entropy of entanglement~\cite{HHHO05,PV07} to obtain Theorem \ref{theorem1}. Let the share of a player $r$ be an ``important'' share if there is an unauthorized set $T$ such that $A=T\cup \{r\}$ is authorized.
\begin{thm}\label{theorem1}
Every ESS satisfies
\begin{equation}\label{thmeq}
E_R(\rho^{\text{DP}})\leq2q
\end{equation}
where $E_R$ is the relative entropy of entanglement, $\rho^{\text{DP}}$ is any state shared between a dealer~D and a set of players~P, and $q$ is the size of the smallest important share (i.e. $q$ is the number of qubits in that share).
\end{thm}
\begin{proof}
Let $\Gamma$ be an unauthorized set satisfying $A = \Gamma\cup \{u\}$ where $A$ is an authorized set and $\{u\}\notin\Gamma$ is the player who possesses the smallest important share, and let $\rho^{\text{D}\Gamma u}$ be the reduced density matrix of joint system of $A$ and the dealer. Assume that the
total dephasing (i.e., twirling) is performed on the smallest important share. That is, the unitary operators $g_i^u\in\{I,X,Y,Z\}^{\otimes q}$ are applied to $u$ with equal probabilities. Then relative entropy of entanglement $E_{R}$ satisfies
\begin{align}
\label{property}
	\sum_i p_i E_R(\rho_i)-&E_R\left(\sum_i p_i\rho_i\right)\nonumber\\
		&\leq S\left(\sum_i p_i\rho_i\right)-\sum_i p_i S(\rho_i)
\end{align}
with
\begin{equation}
	\rho_i=(\one^{\text{D}\Gamma}\otimes g_i^u)\rho^{\text{D}\Gamma u}(\one^{\text{D}\Gamma}\otimes g_i^u)
\end{equation}
and $p_i=\frac{1}{4^q}$.
Due to total dephasing, we have
\begin{equation}
\sum_i p_i\rho_i=\text{tr}_u(\rho^{\text{D}\Gamma u})\otimes\frac{1}{2^k}\one^u.
\end{equation}
Hence, the player $u$ does not now contribute to the entanglement with the dealer because it is independent of the other players and the dealer.
Thus,
\begin{equation} E_R\left(\sum_i p_i\rho_i\right)=E_R(\text{tr}_u(\rho^{\text{D}\Gamma u})).\end{equation}

As the relative entropy of entanglement is invariant under local unitary transformations, we also have
\begin{equation}
\sum_i p_i E_R(\rho_i)=\sum_i p_i E_R(\rho^{\text{D}\Gamma u})=E_R(\rho^{\text{D}\Gamma u}).
\end{equation}
Using the above equations we can rewrite (\ref{property}) as
\begin{align}
	E_R(\rho^{\text{D}\Gamma u})-&E_R(\text{tr}_u(\rho^{\text{D}\Gamma u}))\nonumber\\
		&\leq S(\sum_i p_i\rho_i)-\sum_i p_iS(\rho_i)\nonumber\\
		& \leq H(p_i)=2q.
\end{align}
\indent As $\Gamma$ is an unauthorized set, its reduced density matrix $\text{tr}_u(\rho^{\text{D}\Gamma u})$ is not entangled with the dealer; i.e., $E_R(\text{tr}_u(\rho^{\text{D}\Gamma u}))=0$. On the other hand, $\rho^{\text{D}\Gamma u}$ has the same amount of entanglement as $\rho^{\text{DP}}$ because $\rho^{\text{D}\Gamma u}$ can be transformed into $\rho^{\text{DP}}$ by LOCC and vice versa. Therefore, we have $E_R(\rho^{\text{DP}})=E_R(\rho^{\text{D}\Gamma u})$.    \qedhere
\end{proof}
Let a player $r$ possess a one-qubit important share, where $T$ is an unauthorized set and $A=T\cup \{r\}$ is authorized.
As the set $A$ can recover the entangled state $\rho^{\text{DP}}$ initially shared between a dealer~D and a set of players~P, the
 state $\rho^{\text{DA}}$ has the same amount of entanglement as $\rho^{\text{DP}}$, but the amount of entanglement must go to zero by discarding the share $r$. According to~\cite{HHHO05}, the amount of entanglement can decrease at most by two upon discarding one qubit with respect to the
 relative entropy of entanglement.

Any bipartite entanglement measure has the same value for an MES~\cite{PV07}. Therefore, if a dealer~D shares an MES with a set of players~P, we have $E(\rho^{\text{DP}})=E_R(\rho^{\text{DP}})$. This leads to the following corollary to Theorem \ref{theorem1}.
\begin{cor}
For any entanglement measure $E$,
\begin{equation}E(\rho^{\text{DP}})\leq2q\end{equation}
where $\rho^{\text{DP}}$ is an initial MES and $q$ is the size of the smallest important share.
\end{cor}
We note here a distinct difference between an ESS and a QSS, despite the close analogies of the two processes;
in QSS the size of the secret can be no larger than that of the smallest share \cite{GOT00},
whereas in an ESS
(for which we find protocols that saturate this bound) one can share twice as much entanglement as the size of the smallest share.

\indent Just as in QSS, a natural starting point for constructing an ESS is to use an existing QECC,
as the recovery operation in an ESS is a form of correction of erasure errors (the missing shares of the players which are not part of the subset performing the recovery).  Thus, a given QECC naturally satisfies the necessary recovery condition for some ESS.
However, the secrecy condition for an ESS may be violated if the code
space yields unauthorized sets that share partially entangled states (i.e., neither maximally entangled nor separable) with the dealer.\\
\indent
As $[[n,k,d]]$ stabilizer codes are simple well-understood codes,
we focus on using such codes as the basis for seeking ESSs.
A dealer encodes a MES
\begin{equation}
	|S_\ell\rangle:=\frac{1}{\sqrt{\ell}}\sum_{j=0}^{\ell-1}|j\rangle|j\rangle
\end{equation}
via the mapping
\begin{equation}
	\left|S_{2^k}\right\rangle\mapsto|\Psi\rangle=\frac{1}{\sqrt{2^k}}\sum_{j=0}^{2^k-1}|j\rangle|C_j\rangle
\end{equation}
where the code space $C$ of the $[[n,k,d]]$ stabilizer code is spanned by $\{|C_j\rangle\}_{j=0\cdots 2^k-1}$
with each~$C_j$ a codeword of~$C$.
Each qubit of the code is taken as a share (i.e., $C=\mathscr{H}_{P_1}\otimes\mathscr{H}_{P_2}\otimes\cdots\otimes\mathscr{H}_{P_n}$ where $\mathscr{H}_{P_i}$ is the Hilbert space of the $i^\text{th}$ share).  \\
\indent In this scheme, an access structure is determined from the stabilizer~$S$ of the stabilizer code. Such a code can detect all errors~$E$ that are either in~$S$ or anti-commute with any element of~$S$~\cite{GOT97}.
If the stabilizer code can correct erasure errors on the set~$K$ of shares, the complement of~$K$ would be an authorized set that can reconstruct $|S_\ell\rangle$ through the recovery operation of the code. In this way we can determine an access structure as well as an adversary structure. \\
\indent In order to satisfy the secrecy condition of ESSs, we require that the reduced density matrix of every unauthorized set is separable with the dealer, which will not automatically be the case for a given QECC.  We find, for example, that ESSs can be constructed from Shor's nine-qubit code \cite{Sho95} (i.e., a $[[9,1,3]]$ stabilizer code) and from a $[[4,2,2]]$ stabilizer code~\cite{GOT97}.  In both these schemes the reduced density matrix of every unauthorized set is straightforward to write in a separable form with the dealer.
In general determining whether or not a bipartite density matrix is entangled
is NP-Hard~\cite{GUR03} so verifying the secrecy of an ESS is typically hard.

\subsection{Entanglement sharing using the [[4,2,2]] code}
We now give an explicit example of an ESS.  As shown below, the [[4,2,2]] stabilizer code leads to a threshold ESS (i.e., a scheme in which the access structure is determined solely by the number of players in each subset) due to the permutation invariance of the elements of its stabilizer.
Moreover, this scheme saturates the bound of Eq.~(\ref{thmeq}) as the size of shares is half of the size of initial entanglement.

We consider the case that a dealer~D holds half of $|S_4\rangle$
and encodes the other half into four shares $P=\{1,2,3,4\}$ according to
\begin{equation}\label{encoding4}
	|S_4\rangle\;\rightarrow\;|\Psi\rangle=\frac{1}{2}\sum_{j=0}^3|j\rangle_\text{D}|\beta_j\rangle_{12}|\beta_j\rangle_{34}
\end{equation}
with
\begin{align}
	|\beta_0\rangle=&\frac{|00\rangle+|11\rangle}{\sqrt{2}},\;
	|\beta_1\rangle=\frac{|00\rangle-|11\rangle}{\sqrt{2}},\nonumber\\
	|\beta_2\rangle=&\frac{|01\rangle+|10\rangle}{\sqrt{2}},\;
	|\beta_3\rangle=\frac{|01\rangle-|10\rangle}{\sqrt{2}},
\end{align}
the set of Bell states, where $\text{D}$ denotes the dealer's qubit and the index of $|\beta\rangle_{xy}$ refers to the $x^\text{th}$ and $y^\text{th}$ shares.\\
\indent As the stabilizer code can correct any single erasure error, $|S_4\rangle$ can be recovered from any three or more shares.
Thus, any single share cannot be entangled with the dealer, and we can show explicitly that any single share is in a product state with the dealer, of the form
\begin{equation}
	\frac{1}{4}\sum_{j}|j\rangle_\text{D}\langle j|\otimes\frac{1}{2}\one_{x}
\end{equation}
with $\one_{x}$ the identity matrix of a single share $x$. \\
\indent The subsets $\{1,2\}$ and $\{3,4\}$ each have a separable state with the dealer as is clear from Eq.~(\ref{encoding4}).
The following analysis shows that all other sets of two shares are also separable with the dealer. Note that, as found in \cite{GW10}, the indices of qubits in the state $|\beta_j\rangle_{12}|\beta_j\rangle_{34}$ can be re-ordered as
\begin{equation}
|\beta_j\rangle_{12}|\beta_j\rangle_{34}=\sum_k c_{jk}|\beta_k\rangle_{pq}|\beta_k\rangle_{rs}
\end{equation}
where $\{p,q,r,s\}$ is a permutation of $P=\{1,2,3,4\}$. By tracing out two arbitrary players $\{p,q\}$, one obtains the separable state
\begin{align}
\text{tr}_{pq}\Bigr(|\Psi\rangle\langle \Psi|\Bigl)
& =\frac{1}{4}\sum_k\Bigl(\sum_{j,j'}c_{jk}c_{j'k}|j\rangle_\text{D}\langle j'|\Bigr)\otimes|\beta_k\rangle_{rs}\langle\beta_k|\nonumber\\
& =\frac{1}{4}\sum_k|a_k\rangle_\text{D}\langle a_k|\otimes|\beta_k\rangle_{rs}\langle\beta_k|
\end{align}
where $|a_k\rangle=\sum_j c_{jk}|j\rangle_\text{D}$. \\

\subsection{Entanglement sharing using Shor's code}\label{subshor}
An ESS with a general (not threshold) access structure can be constructed from Shor's nine-qubit code. In this scheme, half of $|S_2\rangle$ is encoded into nine shares $P=\{1,2,3,\dots,9\}$ of 1 qubit, producing the joint state of dealer and players of
\begin{equation}\label{encoding9}
|\Psi\rangle=\frac{1}{\sqrt{2}}\sum_{j=0}^1|j\rangle_\text{D}|G_j\rangle_{123}|G_j\rangle_{456}|G_j\rangle_{789}.
\end{equation}
with
\begin{equation}
|G_0\rangle=\frac{1}{\sqrt{2}}(|000\rangle+|111\rangle)
\end{equation}
and
\begin{equation}
|G_1\rangle=\frac{1}{\sqrt{2}}(|000\rangle-|111\rangle).
\end{equation}
We note that the encoding is not fully symmetric between shares, but divides the players into three ``triplets'' $\{\{1,2,3\}, \{4,5,6\}, \{7,8,9\}\}$.

The nine-qubit Shor code~\cite{Sho95}
can correct erasure errors on any of those subsets $K\subseteq P$ which consist of
\begin{enumerate}
 \item any one or two qubits,
 \item two qubits in one triplet and a single qubit in another triplet (e.g., $\{1,2,4\}$), or
 \item two qubits in one triplet and two qubits in another triplet (e.g., $\{1,2,4,5\}$).
\end{enumerate}
Therefore, the complement of $K$ can recover $|S_2\rangle$ using error correction for the 9-qubit code, and $K$ cannot be entangled with the dealer due to monogamy. \\
\indent The remaining subsets of players $\{B\}$, consisting neither of subsets $K$ nor their complements, are separable with the dealer.
For example, the reduced density matrix of $\{1,2,3,4,5,6\}$ can be written as
\begin{align}\label{789}
\text{tr}_{789}&\Bigl(|\Psi\rangle\langle\Psi|\Bigr)\nonumber\\
&=\frac{1}{2}\sum_{j=0}^{1}|j\rangle\langle j|_\text{D}\otimes|G_j\rangle_{123}\langle G_j|\otimes|G_j\rangle_{456}\langle G_j|.
\end{align}
In a similar way the other subsets in $\{B\}$ are also separable with the dealer, as shown in the Appendix \ref{app1}. \\

\subsection{Non-perfect entanglement sharing}
\indent As shown below, not all stabilizer codes satisfy the secrecy condition for standard entanglement sharing, although we can also consider them as a ``non-perfect'' case, in which we allow some partial entanglement between unauthorized subsets and the dealer.  For example, a [[6,4,2]] stabilizer code~\cite{GOT97} yields a non-perfect ESS. Suppose that a dealer encodes an MES of four ebits by the $[[6,4,2]]$ code such that
\begin{equation}
\frac{1}{4}\sum_{i=1}^{16}|i\rangle|i\rangle\mapsto \frac{1}{4}\sum_{i=1}^{16}|i\rangle|\beta_{f(i)}\rangle|\beta_{g(i)}\rangle|\beta_{f(i)+g(i)}\rangle
\end{equation}
where $f(i),\,g(i)=\{0,1,2,3\}$ and $\{|\beta_j\rangle\}$ are the Bell states.\\
\indent Any five players can recover the original entanglement as this stabilizer code can correct erasure errors on a single qubit. However, any four players have some entanglement with the dealer even though they are an unauthorized set, as can be seen in two ways. First the reduced density matrix of four players violates the positive partial transpose criterion~\cite{PER96,HHH96}, which provides one useful way to check state separability in ESSs when separability is not obvious from the density matrices. Second, the amount of initially shared entanglement is too large to satisfy Theorem \ref{theorem1} (i.e., $E_R(\rho^{\text{DP}})=4$ and $q=1$). Therefore, some entanglement must remain even after one player is excluded from an authorized set.  \\

\section{Hybrid Entanglement Sharing}

\indent A QECC that does not already satisfy the secrecy condition can potentially be made to do so through hybridization with classical information.  Similarly to hybrid QSS \cite{NMI01,FG12},
hybrid ESSs can be implemented from any QECC combined with classical secret sharing, by ``locking'' any leaked entanglement from recovery by first encrypting the dealer's MES using classical keys (i.e., bit strings) and then distributing the keys among a set of players in such a way that unauthorized sets are totally denied any access to entanglement. \\
\indent The principle of encrypting quantum information with classical keys is shown in quantum teleportation~\cite{BBC93}. Alice generates two classical bits by performing a joint measurement on a quantum state $|\psi\rangle$ that she intends to teleport to Bob and half of a previously shared MES. In the absence of Alice's classical bits, Bob's qubit is left in a maximally mixed state,
\begin{equation}
\frac{1}{4}(|\psi\rangle\langle\psi|+X|\psi\rangle\langle\psi|X+Z|\psi\rangle\langle\psi|Z+XZ|\psi\rangle\langle\psi|ZX).
\end{equation}
However, with the knowledge of the classical bits, Bob can recover $|\psi\rangle$ by determining which state of the above mixture his qubit is in. In a similar sense we can encrypt and decrypt entanglement using classical keys. \\

Let us consider an MES
\begin{equation}
	|E\rangle=\frac{1}{\sqrt{q}}\sum_{j=0}^{q-1}|a_j\rangle|b_j\rangle
\end{equation}
in a bipartite system of two $q$-dimensional Hilbert spaces. We can partially encrypt $|E\rangle$ by a unitary mapping
\begin{equation}
U^l:|E\rangle\mapsto|E^l\rangle=\frac{1}{\sqrt{q}}\sum_{j=0}^{q-1}\omega_{jl}|a_j\rangle|b_j\rangle
\end{equation}
with
\begin{equation}
	\omega_{jl}=\exp\left(\frac{2i\pi jl}{q}\right),\;l\in\mathbf{Z}_{q},
\end{equation}
randomly chosen. Note that the phase information for the encrypted state $|E^l\rangle$ is totally randomized for any party without knowledge of the classical key $l$, and thus the randomized state can be written in terms of separable states as follows:
\begin{align}
	\rho^S &=\sum_{l=0}^{q-1}\left|E^l\right\rangle\left\langle E^l\right|\nonumber\\
&=\frac{1}{q}\sum_{j,j'=0}^{q-1}\delta(j-j')|a_j\rangle\langle a_{j'}|\otimes |b_j\rangle\langle b_{j'}|\nonumber\\
&=\frac{1}{q}\sum_{j=0}^{q-1}|a_j\rangle\langle a_j|\otimes|b_j\rangle\langle b_j|.
\end{align}
\indent Hybrid ESSs can be implemented by the following procedure: Suppose that a dealer encodes an MES $|S\rangle$ into the code space $C=\text{span}\{C_j\}_{j=0\cdots q-1}$ of a QECC and thereby obtains
\begin{equation}
|\Psi\rangle=\frac{1}{\sqrt{q}}\sum_{j=0}^{q-1}|j\rangle|C_j\rangle.
\end{equation}
The access structure $Q$ for the corresponding ESS is determined by the code's ability to recover from erasure errors (for example an $[[n,k,d]]$ stabilizer code corresponds to an access structure $Q$ consisting of any $n-d+1$ or more shares).  The dealer additionally performs the following steps:
\begin{enumerate}
 \item Performs $U^l$ on $|\Psi\rangle$ using a randomly chosen classical key $l$.
 \item Encodes $l$ in classical shares and distributes using a classical secret sharing scheme with an access structure $Q'\subseteq Q$.
\end{enumerate}
We thus have a scheme with both classical and quantum shares, which are separately distributed among the $n$ players. If a set of players is an element of $Q'$, it can recover $|S\rangle$ with the classical key $l$. Otherwise they are left in a separable state with a dealer because they cannot acquire any information about $l$. We note that suitable classical secret sharing schemes can be easily devised by using polynomial functions~\cite{Sha79,Bla79}.\\

\section{Describing partial information in QRSS}

\indent  As described earlier, we can consider both the ``perfect'' ESS where unauthorized sets are denied all entanglement and the ``imperfect'' case where this secrecy condition is relaxed and some entanglement is leaked.  A similar framework applies to leaked information in imperfect QSS, known as quantum ramp secret sharing (QRSS).
In QRSS, one considers three types of player structures: access, forbidden and intermediate structures. The forbidden structure is the collection of unauthorized sets that are completely denied any information about the secret, and the intermediate structure is the collection of unauthorized sets that obtain some information about a secret.

As QRSS allows for smaller player shares than perfect QSS, this information leakage could be an acceptable sacrifice,
but, for a given secret sharing scenario, what constitutes acceptable leakage likely depends on the nature of the information that is accessible to the intermediate subsets.
One method to characterize information leakage requires stabilizer encoding
and therefore requires precise details of encoding operation (e.g., stabilizers) to generate the information group~\cite{GLG10}
so is not of direct use here.

\indent Entanglement sharing gives us an alternative way to characterise this information, which may be useful in many circumstances.  In a perfect ESS,  we can similarly consider three structures of players with respect to correlations with the dealer, by dividing the adversary structure into a forbidden structure and an intermediate structure.  Whereas all members of the adversary structure recover no entanglement with the dealer, we can distinguish between forbidden subsets (who can recover only product states with the dealer) and intermediate subsets, who can recover non-product mixed states with the dealer; i.e., they and the dealer can share {\it classical correlations} (but no entanglement).  Thus player subsets can be divided into those recovering ``quantum correlations'' (entanglement), classical correlations, or no correlations with the dealer.
This can be a useful qualitative description of leaked information.

One can characterise the qualitative difference by considering, for example, a circuit implementing quantum teleportation~\cite{BBC93}, which outputs a density matrix associated with an input state. Given an MES as a resource, this circuit can be successfully operated and thus the output density matrix is identical to the input state. However, if a product state is instead used as the resource, the circuit outputs the identity matrix. If a separable state is used as a resource, quantum teleportation cannot be perfectly achieved but will output a density matrix with some correlation with the input state.

For example, suppose that an arbitrary quantum state $|\phi_{S}\rangle=\alpha|0\rangle+\beta|1\rangle$ is the input state of the quantum teleportation circuit, and a separable state
\begin{equation}
\frac{1}{2}\Bigl(|0\rangle\langle0|\otimes|0\rangle\langle0|
+|1\rangle\langle1|\otimes|1\rangle\langle1|\Bigr)
\end{equation}
is used as its resource. Then the circuit outputs
\begin{equation}
\rho_{R}=\lvert\alpha\rvert^2|0\rangle\langle0|+\lvert\beta\rvert^2|1\rangle\langle1|.
\end{equation}
The output density matrix of this ``classical teleportation'' contains some information about the diagonal elements of the density matrix $\rho_{S}=|\phi_{S}\rangle\langle\phi_{S}|$ but nothing about its off-diagonal elements. In this context we can characterize leaked information about the secret in a QRSS scheme as ``quantum'' or ``classical'' in the sense of requiring a quantum or only a classical channel to transmit.\\
\indent By considering QRSS schemes as ESSs (applying the scheme to a secret of half of an MES, with the other half retained by the dealer), we therefore have a means of characterizing this information: a unauthorized subset of players who can recover a partially-entangled state with the dealer have some leaked quantum information, whereas those recovering a separable state only have leaked classical information.  This method of classifying information in the context of quantum teleportation provides one unambiguous definition of leakage of specifically quantum information, as it relies only on the clear definitions of entangled and separable states.  Since quantum correlations are useful resources in many contexts, a QRSS scheme may be acceptable if, for example, the only information leaked is classical.\\
\indent Now let us revisit the $[[4,2,2]]$ stabilizer code. The $[[4,2,2]]$ stabilizer code is not only used to devise an ESS but also for a $(3,2,4)$ QRSS scheme~\cite{GLG10} (where $(3,2,4)$ denotes a four-player scheme in which 3-player subsets are authorized, whereas $3-2=1$ player subsets receive no information about the secret, and two-player subsets receive partial information). Suppose that the ESS using the $[[4,2,2]]$ stabilizer code is applied to supply an MES required for quantum teleportation. If an arbitrary quantum state $\rho_\text{D}$ is transmitted to four players, any three or more players can receive the state perfectly with the recovered MES, and any single player can acquire no information. However, any two players can have a separable state with the dealer and learn the information about diagonal elements of $\rho_\text{D}$. Therefore, the $[[4,2,2]]$ stabilizer code yields a QRSS scheme that leaks only classical information to unauthorized sets, which corresponds to the result of~\cite{GLG10}.\\
\indent Formally, we propose a secrecy condition for QRSS in terms of entanglement sharing as follows.
\begin{defn}
Given an encoding operation $C:\mathscr{H}_\text{D}\rightarrow\mathscr{H}_\text{P}$ mapping
\begin{equation}
\begin{split}
(\one_\text{D}\otimes C) : |S\rangle=\frac{1}{\sqrt{|C|}}\sum_{j}|j\rangle_\text{D}|j\rangle_\text{D}\rightarrow \\ |\Psi\rangle=\frac{1}{\sqrt{|C|}}\sum_{j}|j\rangle_\text{D}|C_j\rangle_\text{P},
\end{split}
\end{equation}
a QRSS scheme described by $C$ is secure from the leakage of quantum information if the reduced density matrix for every intermediate set $\Gamma\in P$, $$\rho^\Gamma=\text{tr}_{\bar{\Gamma}}(|\Psi\rangle\langle\Psi|),$$ is separable with the system $\text{D}$ for $\bar{\Gamma}$ the complement of~$\Gamma$.
\end{defn}

\section{Conclusions}

\indent We have introduced a new protocol for entanglement sharing, which allows a dealer to distribute half of a MES with a set of players in such a way that some collaborating groups of players can recover the entangled state fully, but other groups cannot share any entanglement with the dealer.  While closely related to QSS (every perfect QSS scheme will also be an entanglement sharing scheme) the entanglement sharing conditions result in different properties for these schemes (for example share size of each player need only be at least half the amount of initial entanglement).  We have demonstrated examples of construction of entanglement sharing schemes from stabilizer codes, both directly and through hybridization with classical encryption.  We note that the general relationship between QECCs and entanglement sharing (i.e., whether or not schemes can be directly constructed from a given QECC) is still not evident,
and therefore is a promising avenue for further investigation.

Finally, we have shown that the entanglement sharing paradigm provides a useful characterization of information leakage in QRSS schemes, wherein any QECC suitable for perfect entanglement sharing can be used to construct a QRSS in which only classical information is leaked to unauthorized player subsets. As with sharing quantum secrets, entanglement sharing would be especially convenient if one could choose the access, intermediate, and adversary structure first and then find a corresponding code, but in practice the codes are chosen first and the structures are consequential. Given the many contexts in which entanglement is a crucial resource for performing quantum information protocols, this characterization, and entanglement sharing schemes in general, have the potential for a wide range of applications.

\begin{acknowledgments}
We appreciate valuable discussions with Vlad Gheorghiu and Yun-jiang Wang.
This project has been supported by Alberta Innovates Technology Futures, AITF, CIFAR and NSERC.
\end{acknowledgments}

\appendix*
\section{Investigation of Shor's code}\label{app1}
In this appendix, we supplement Sec. \ref{subshor} by investigating other sets in $\{B\}$.
As $\{1,2,3,4,5,6\}$ is in a separable state with a dealer, $\{1,2,3,7,8,9\}$ and $\{4,5,6,7,8,9\}$ are also separable with
the dealer as shown by permutation of the triplets. Similarly the complement of $\{1,2,3,4,5,6\}$, namely $\{7,8,9\}$, is in a separable state
\begin{equation}
\frac{1}{2}\sum_{j=0}^{1}|j\rangle_\text{D}\langle j|\otimes|G_j\rangle_{789}\langle G_j|
\end{equation}
and thus $\{1,2,3\}$ and $\{4,5,6\}$ are also separable with the dealer. We regard these six subsets as a class of $\{1,2,3\}$. \\
\indent Every set in $\{B\}$ is classified into four classes: $\{1,2,3\}$, $\{1,4,7\}$, $\{1,2,3,4\}$ and $\{1,4,7,8\}$. For each class, it is enough to demonstrate whether or not the
reduced density matrices of a small subset (e.g., $\{1,2,3\}$) and a big subset (e.g., $\{1,2,3,4,5,6\}$) are separable with the dealer. \\
\indent Let us look at the class of $\{1,4,7\}$. The complement of $\{1,4,7\}$ has a separable state written as
\begin{equation}\label{147}
\frac{1}{2}\Bigl(|+\rangle_\text{D}\langle +|\otimes \rho_{235689}+ |-\rangle_\text{D}\langle -|\otimes\rho'_{235689}\Bigr)
\end{equation}
with
\begin{align}
	|{\pm}\rangle_\text{D}:=&\frac{1}{\sqrt{2}}(|0\rangle_\text{D}\pm|1\rangle_\text{D}),\nonumber\\
	\rho_{235689}=&\frac{1}{4}\Bigl(|000000\rangle\langle000000|+|001111\rangle\langle001111|\nonumber\\
&+|110011\rangle\langle110011|+|111100\rangle\langle111100|\Bigr),\nonumber\\
\rho'_{235689}&=\frac{1}{4}\Bigl(|000011\rangle\langle000011|+|001100\rangle\langle001100|\nonumber\\
&+|110000\rangle\langle110000|+|111111\rangle\langle111111|\Bigr).
\end{align}

If we discard three more qubits $\{2,5,8\}$ from the state in Eq.~(\ref{147}), the resulting density matrix is a reduced density matrix of $\{3,6,9\}$ in this class and still has a separable form with the dealer.

Next we consider the class of $\{1,2,3,4\}$. The reduced density matrix for $\{5,6,7,8,9\}$ (the complement of $\{1,2,3,4\}$) is given by
\begin{align}\label{56789}
\frac{1}{2}\Bigl(\sum_{j=0}^{1}|j\rangle_\text{D}\langle j|&\otimes|G_j\rangle_{789}\langle G_j|\Bigr)\nonumber\\
&\otimes\frac{1}{2}\Bigl(|00\rangle_{56}\langle00|+|11\rangle_{56}\langle11|\Bigr).
\end{align}
In this case, the subset $\{5,6\}$ is independent of the other qubits so the reduced density matrix for $\{6,7,8,9\}$ (a small set in the class of $\{1,2,3,4\}$) is separable with the dealer after tracing out player 5. \\
\indent Finally we consider the class of $\{1,4,7,8\}$. By tracing out $\{8\}$ and $\{5,8\}$ from the state of Eq. (\ref{147}),
clearly $\{2,3,5,6,9\}$ and $\{2,3,6,9\}$ in this class are separable with the dealer.
We therefore conclude that any subset of players in $\{B\}$  is in a separable state with the dealer.

\bibliography{ppbib}

\end{document}